\documentclass[envcountsame,fleqn]{llncs}
 
\usepackage{style1}

\title{Descriptive Complexity of $\cAC$ Functions%
\thanks{Partially supported by DFG VO 630/8-1, grant  292767 of the Academy of Finland and grant ANR-14-CE25-0017-02 AGREG of the ANR}}

\author{Arnaud Durand\inst{1} \and Anselm Haak\inst{2} \and Juha Kontinen\inst{3} \and Heribert Vollmer\inst{2}}
\institute{Université Paris Diderot, IMJ-PRG, CNRS UMR 7586, Case 7012, 75205 Paris cedex 13, France\\\texttt{durand@logique.jussieu.fr}
\and Theoretische Informatik, Leibniz Universität Hannover, Appelstraße, D-30167, Germany\\\texttt{(haak|vollmer)@thi.uni-hannover.de}
\and Department of Mathematics and Statistics, University of Helsinki, Finland\\\texttt{juha.kontinen@helsinki.fi}}

\begin{document}
\maketitle

\pagestyle{plain}

\begin{abstract}
We introduce a new framework for a descriptive complexity approach to arithmetic computations. We define a hierarchy of classes based on the idea of counting assignments to free function variables in first-order formulae. We completely determine the inclusion structure and show that \cP and \cAC appear as classes of this hierarchy. In this way, we unconditionally place \cAC properly in a strict hierarchy of arithmetic classes within \cP. We compare our classes with a hierarchy within \cP defined in a model-theoretic way by Saluja et al. We argue that our approach is better suited to study arithmetic circuit classes such as \cAC which can be descriptively characterized as a class in our framework.
\end{abstract}

\section{Introduction}

The complexity of arithmetic computations is a current focal topic in  complexity theory. Most prominent is Valiant's class \cP of all functions that count accepting paths of nondeterministic polynomial-time Turing machines. This class has interesting complete problems like counting the number of satisfying assignments of propositional formulae or counting the number of perfect matchings of bipartite graphs (the so called permanent \cite{val79b}). 

The class \cP has been characterized in a model-theoretic way by Saluja, Subrahmanyam and Thakur in \cite{DescCompNumP}. Their characterization is a  natural generalization of Fagin’s Theorem: Given a first-order formula with a free relational variable, instead of asking if there exists an assignment to this variable that makes the formula true (NP $=$ ESO), we now ask to count how many assignments there are. In this way, the class \cP is characterized: $\cP = \cFO^\rel$. We use the superscript $\rel$ to denote that we are counting assignments to relational variables.

From another point of view, the class \cP can be seen as the class of those functions that can be computed by arithmetic circuits of polynomial size, i.e., circuits with plus and times gates instead of the usual Boolean gates (cf., e.g., \cite{HeribertBuch}). This is why here we speak of \emph{arithmetic computations}. 

It is very natural to restrict the resource bounds of such arithmetic circuits. An important class defined in this way is the class $\cAC$ of all functions computed by polynomial-size bounded-depth arithmetic circuits. It is interesting to note that \cAC and all analogous classes defined by arithmetic circuits, i.e., plus-times circuits, can also be defined making use of a suitable counting process: A witness that a Boolean circuit accepts its input is a so called proof tree of the circuit, i.e., a minimal subtree of the circuit unwound into a tree, in which all gates evaluate to $1$. Then the arithmetic class \cAC can be characterized as the counting class of all functions that count proof trees of \AC circuits.

There was no model-theoretic characterization of \cAC, until it was recently shown in \cite{HaVo16} that $\cAC = \cPiSk$, where $\cPiSk$ means counting of possible Skolem functions for \FO-formulae.

The aim of this paper is to compare the above two model-theoretic characterizations  in order to get a unified view for both arithmetic circuit classes, \cAC and \cP. This is done by noticing that the number of Skolem functions of an \FO-formula can be counted as satisfying assignments to free function variables in a $\Pi_1$-formula. This gives rise to the idea to restate the result by Saluja et al counting functions instead of relations. We call our class where we count assignments to function variables \cFO, in contrast to Saluja et al.'s $\cFO^\rel$. In this setting, we get $\cP = \cFO = \cPii{1}$, which places both classes within $\cPii{1}$.

Furthermore, we show that $\cAC$  actually corresponds to a syntactic fragment $\cPiPrefix$ of $\cPii{1}$ and, considering further syntactic subclasses of \cFO defined by quantifier alternations, we get the inclusions
\begin{equation}\label{eqInclusions}
\cSigmai{0} 
\begin{array}{@{\ }r@{\ }c@{\ }l@{\ }}
\rotatebox[origin=c]{30}{$\subsetneq$} 
& \raisebox{.75ex}{\cAC = \cPiPrefix} 
& \rotatebox[origin=c]{-30}{$\subsetneq$} 	\\
\rotatebox[origin=b]{-20}{$\subsetneq$} 
& \raisebox{-.75ex}{\cSigmai{1}} 
& \rotatebox[origin=c]{20}{$\subsetneq$}	
\end{array}
\cPii{1} = \cFO = \cP
\end{equation}
Thus we establish (unconditionally, i.e., under no complexity theoretic assumptions) the complete structure of the alternation hierarchy within \cFO and show where \cAC is located in this hierarchy. 

Once we know that only universal quantifiers suffice to obtain the full class, i.e., $\cPii{1}=\cP$, it is a natural question to ask how many universal quantifiers are needed to express certain functions. We obtain the result that the hierarchy based on the number of universal variables is infinite; however, a possible connection to the depth hierarchy within \cAC remains open.

We see that counting assignments to free function variables instead of relation variables in first-order formulae leads us to a hierarchy of arithmetic classes suitable for a study of the power and complexity of the class \cAC. The hierarchy introduced by Saluja et al.{} \cite{DescCompNumP} does not seem suitable for such a goal.

This paper is organized as follows: In the next section, we introduce relevant concepts from finite model theory. Here, we also introduce the Saluja et al.{} hierarchy, and we explain the model-theoretic characterization of \cAC. In Sect.~\ref{sect-framework} we introduce our new framework and the class \cFO and its subclasses. In Sect.~\ref{sect-hierarchy} we determine the full structure of the alternation hierarchy within \cFO and place \cAC in this hierarchy, while in Sect.~\ref{sect-varhierarchy} we study the hierarchy defined by the number of universal variables in the $\cPii{1}$-fragment. 
In Sect.~\ref{sect-connection} we turn to the hierarchy defined by Saluaj et al.{} and show that the arithmetic class \cAC is incomparable to all except the level-0 class and the full class of this hierarchy. Finally, we conclude in Sect.~\ref{sect-conclusion} with some open questions.

Our proofs make use of a number of different results and techniques, some stemming from computational complexity theory (such as separation of Boolen circuit classes or the time hierarchy theorem for nondeterministic RAMs), some from model theory (like closure of certain fragments of first-order logic under extensions or taking substructures) or descriptive complexity (correspondence between time-bounded NRAMs and fragments of existential second-order logic). Most techniques have to be adapted to work in our very low complexity setting (new counting reductions, use of the right set of built-in relations, etc.). Our paper sits right in the intersection of finite model theory and computational complexity theory.

\section{Definitions and Preliminaries}

In this paper we consider finite $\sigma$-structures where $\sigma$ is a finite vocabulary consisting of relation and constant symbols. For a structure $\calA$, $\textrm{dom}(\calA)$ denotes its universe. We will always use structures with universe $\{0, 1, \dots, n-1\}$ for some $n \in \mathbb{N}\setminus\{0\}$. Sometimes we will assume that our structures contain certain \emph{built-in relations and constants}, e.g., $\leq^2$, $\textrm{SUCC}^2$, $\textrm{BIT}^2$ and $\min$. In the following, we will always make it clear what built-in relations we allow. The interpretations of built-in symbols are fixed for any size of the universe as follows: $\leq^2$ is the $\leq$-relation on $\mathbb{N}$, $min$ is 0, $\textrm{SUCC}(i,j)$ is true, iff $i+1=j$, and $\textrm{BIT}(i,j)$ is true, iff the $i$'th bit of the binary represention of $j$ is 1. We will generally write $\enc_\sigma(\mathcal{A})$ for the binary encoding of a $\sigma$-structure $\mathcal{A}$. For this we assume the standard encoding (see \eg \cite{ImmermanBuch}): Relations are encoded row by row by listing their truth values as $0$'s and $1$'s. Constants are encoded by the binary representation of their value and thus a string of length $\lceil \log_2(n)\rceil$. A whole structure is encoded by the concatenation of the encodings of its relations and constants except for the built-in numerical predicates and constants: These are not encoded, because they are fixed for any input length. 

Since we want to talk about languages accepted by Boolean circuits, we will need the vocabulary
\[\tString = (\leq^2, S^1)\]
of binary strings. A binary string is represented as a structure over this vocabulary as follows: Let $w \in \{0,1\}^*$ with $|w| = n$. Then the structure representing this string is the structure with universe $\{0,\dots,n-1\}$, $\leq^2$ interpreted as the $\leq$-relation on the natural numbers and $x \in S$, iff the $x$'th bit of $w$ is 1. The structure corresponding to string $w$ is denoted by $\mathcal{A}_w$.

For any $k$, the fragments $\Sigma_k$ and $\Pi_k$ of \FO are the classes of all formulae in prenex normal form with a quantifier prefix with $k$ alternations starting with an existential or an universal quantifier, respectively. 

A concept connected to \FO that we will need to define uniformity of circuit families are \FO-interpretations, which are mappings between structures over different vocabularies.

\begin{definition}
Let $\sigma, \tau$ be vocabularies,
$\tau = \langle R_1^{a_1}, \dots, R_r^{a_r}\rangle$. A first-order interpretation (or FO-interpretation)
\[I: \struc[\sigma] \rightarrow \struc[\tau]\]
is given by a tuple of FO-formulae $\varphi_0, \varphi_1, \dots, \varphi_r$ over the vocabulary $\sigma$. For some $k$, $\varphi_0$ has $k$ free variables and $\varphi_i$ has $k \cdot a_i$ free variables \fa $i \geq 1$. For each structure $\mathcal{A} \in \struc[\sigma]$, these formulae define the structure
\[I(\mathcal{A}) = \langle |I(\mathcal{A})|, R_1^{I(\mathcal{A})}, \dots, R_r^{I(\mathcal{A})}\rangle \in \struc[\tau],\]
where the universe is defined by $\varphi_0$ and the relations by $\varphi_1, \dots, \varphi_r$ in the following way:
\[|I(\mathcal{A})| = \{\langle b^1, \dots, b^k\rangle \mid \mathcal{A} \vDash \varphi_0(b^1, \dots, b^k)\}\]
\[R_i^{I(\mathcal{A})} = \{(\langle b_1^1, \dots, b_1^k\rangle, \dots, \langle b_{a_i}^1, \dots, b_{a_i}^k\rangle) \in |I(\mathcal{A})|^{a_i} \mid \mathcal{A} \vDash \varphi_i(b_1^1, \dots, b_{a_i}^k)\}\]
\end{definition}


We will now define the class \cP and a model-theoretic framework in which the class can be characterized. Here, we follow \cite{DescCompNumP} only changing the name slightly to emphasize that we are counting relations in this setting.

\begin{definition}
A function $f: \{0,1\}^* \rightarrow \mathbb{N}$ is in \cP, if there is a non-deterministic Turing-machine $M$ \stfa inputs $x \in \{0,1\}^*$,
\[f(x) = \textrm{number of accepting computation paths of } M \textrm{ on input } x.\]
\end{definition}

\begin{definition}\label{def-forel}
A function $f: \{0,1\}^* \rightarrow \mathbb{N}$ is in $\cFO^\rel$, if there is a vocabulary $\sigma$ including built-in linear order $\leq$, and an \FO-formula $\varphi(R_1, \dots, R_k, x_1, \dots, x_\ell)$ over $\sigma$ with free relation variables $R_1, \dots, R_k$ and free individual variables $x_1, \dots, x_\ell$ \stfa $\calA \in \struc[\sigma]$,
\[f(\enc_\sigma(\calA)) = |\{(S_1, \dots, S_k, c_1, \dots, c_\ell) \mid \calA \vDash \varphi(S_1, \dots, S_k, c_1, \dots, c_\ell\}|.\]
\end{definition}

In the same fashion we define counting classes using fragments of \FO, such as $\cSigmai{k}^\rel$ and $\cPii{k}^\rel$ for arbitrary $k$. In \cite{DescCompNumP} the following was shown for these classes (assuming order as the only built-in relation):

\begin{theorem}\label{thm-saluja-hierarchy}
$\cSigmaiRel{0} = \cPiiRel{0} \subset \cSigmaiRel{1} \subset \cPiiRel{1} \subset \cSigmaiRel{2} \subset \cPiiRel{2} = \cFO^\rel = \cP$.
\end{theorem}

Besides this theorem, it was also shown that the functions in $\cSigmaiRel{0}$ can be computed in polynomial time.

To illustrate the just given definition, we repeat an example from Saluja et al.{} \cite{DescCompNumP} that will also be important for us later.	

\begin{example}\label{example-3DNF}
We will show that $\#3\mathrm{DNF}$, the problem of counting the number of satisfying assignments of a propositional formula in disjunctive normal-form with at most 3 literals per disjunct, is in the class $\cSigmaiRel{1}$. To do so, we use the vocabulary $\sigma_{\#3\mathrm{DNF}}=(D_0,D_1,D_2,D_3)$. Given a 3DNF-formula $\varphi$ over variables $V$, we construct a corresponding $\sigma$-structure ${\cal A}_\varphi$ with universe $V$ such that for any $x_1,x_2,x_3\in V$, $D_i(x_1,x_2,x_3)$ holds iff $\neg x_1\wedge\dots\wedge\neg x_i\wedge x_{i+1}\wedge\dots\wedge x_3$ appears as a disjunct.
Now consider the following $\sigma$-formula with free relational variable $T$:
$$\begin{array}{r@{\ }l}
\Phi_\cDnf(T) = \exists x\exists y\exists z \quad \smash{\Bigl(}
		& \phantom{\vee}\bigl(D_0(x,y,z) \wedge T(x) \wedge T(y) \wedge T(z)\bigr)\\
		& \vee \bigl(D_1(x,y,z) \wedge \neg T(x) \wedge T(y) \wedge T(z)\bigr)\\
		& \vee \bigl(D_2(x,y,z) \wedge \neg T(x) \wedge \neg T(y) \wedge T(z)\bigr)\\
		& \vee \bigl(D_3(x,y,z) \wedge \neg T(x) \wedge \neg T(y) \wedge \neg T(z)\bigr) \smash{\Bigr)}\\
\end{array}
$$
Observe that $\Phi_\cDnf$ is a $\Sigma_1$-formula. 
Evaluated on an input structure $\calA_\varphi$, it expresses that an assignment to $T$ defines a satisfying assignment of $\varphi$. Hence, the number of assignments $\mathbf T$ such that ${\cal A}_\varphi \models \Phi_\cDnf({\mathbf T})$ is equal to the number of satisfying assignments of $\varphi$. 
\end{example}

We will next recall the definition of Boolean circuits and counting classes defined using them. A circuit is a directed acyclic graph (dag), whose nodes (also called gates) are marked with either a Boolean function (in our case $\land$ or $\lor$), a constant (0 or 1), or a (possibly negated) bit of the input. Also, one gate is marked as the output gate. On any input, a circuit computes a Boolean function by evaluating all gates according to what they are marked with. The value of the output gate then is the function value for that input.\\
When we want circuits to work on different input lengths, we have to consider families of circuits: A family contains a circuit for any input length $n \in \mathbb{N}$. Families of circuits allow us to talk about languages beeing accepted by circuits: A circuit family $\mathcal{C} = (C_n)_{n \in \mathbb{N}}$ is said to accept (or decide) the language $L$, if it computes its characteristic function $c_L$:
\[C_{|x|}(x) = c_L(x) \textrm{ for all } x.\]
The complexity classes in circuit complexity are classes of languages that can be decided by circuit families with certain restrictions to their depth and  size. The depth here is the length of a longest path from any input gate to the output gate of a circuit and the size is the number of non-input gates in a circuit. Depth and size of a circuit family are defined as functions accordingly.\\
Above, we have not restricted the computability of the circuit $C_{|x|}$ from $x$ in any way. This is called non-uniformity, which allows such circuit families to even compute non-recursive functions. Since we want to stay within \cP, we need some notion of uniformity. For this, we first define the vocabulary for Boolean circuits as \FO-structures:
\[\tCirc = (E^2, G_\land^1, G_\lor^1, B^1, r^1),\]
where the relations are interpreted as follows:
\begin{itemize}
  \item $E(x,y)$:  $y$ is a child of $x$
  \item $G_\land(x)$: gate $x$ is an and-gate
  \item $G_\lor(x)$: gate $x$ is an or-gate
  \item $B(x)$: gate $x$ is a true leaf of the circuit
  \item $r(x)$: $x$ is the root of the circuit
\end{itemize}

We will now define \FO-uniformity of Boolean circuits in general and the class \FO-uniform \AC.

\begin{definition}
A circuit family $\mathcal{C} = (C_n)_{n \in \mathbb{N}}$ is said to be first-order uniform (FO-uniform) if there is an FO-interpretation
\[I: \struc[\tString\cup(\textrm{BIT}^2)] \rightarrow \struc[\tCirc]\]
mapping any structure $\mathcal{A}_w$ over \tString to the circuit $C_{|w|}$ given as a structure over the vocabulary \tCirc. 
\end{definition}

\begin{definition}
A language $L \subseteq \{0,1\}^*$ is in \FO-uniform \AC, if there is an \FO-uniform circuit family with constant depth and polynomial size accepting $L$.
\end{definition}

It is known that the just given class coincides with the class $\FO$ of all languages definable in first-order logic \cite{baimst90,ImmermanBuch}, i.e., informally: $\FO\text{-uniform }\AC = \FO$. For this identity, it is central that our logical language includes the built-in relations of linear order and BIT. 

We will next define counting classes corresponding to Boolean circuit families. For a nondeterministic Turing machine, the witnesses we want to count are the accepting paths of the machine on a given input. Considering polynomial time computations, this concept gives rise to the class \cP. A witness that a Boolean circuit accepts its input is a so called \emph{proof tree}, a minimal subtree of the circuit showing that it evaluates to true for a given input. For this, we first unfold the given circuit into tree shape, and we further require that it is in \emph{negation normal form} (meaning that negations only occur directly in front of literals). A proof tree then is a subtree we get by choosing for any $\lor$-gate exactly one child and for any $\land$-gate all children, \ST every leaf which we reach in this way is a true literal. This allows us to define the class \cAC as follows:

\begin{definition}
A function $f: \{0,1\}^* \rightarrow \mathbb{N}$ is in \FO-uniform \cAC, if there is an \FO-uniform circuit family $\mathcal{C} = (C_n)_{n \in \mathbb{N}}$ such that for all $w \in \{0,1\}^*$,
\[f(w) = \textrm{number of proof trees of } C_{|w|}(w).\]
\end{definition}


As was shown in \cite{HaVo16}, there is a model-theoretic characterization of \FO-uniform \cAC. For this, let us define the Skolemization of an \FO-formula $\varphi$ in prenex normal form by removing all existential quantifiers and replacing each existentially quantified variable in the matrix of $\varphi$ by a term consisting of a function application to those variables quantified universally to the left of the original existential quantifier.
In other words, every existential variable is replaced by its so called \emph{Skolem function}. Now, \cAC contains exactly those functions that can be given as the number of Skolem functions for a given \FO-formula.

\begin{definition}\label{def-pisk}
A function $f\colon \{0,1\}^* \rightarrow \mathbb{N}$ is in the class $\cPiSk$ if there is a vocabulary $\sigma$ including built-in linear order, BIT and min and a first-order sentence $\varphi$ over $\sigma$ in prenex normal form  
\[\varphi \triangleq \exists y_1 \forall z_1 \exists y_2 \forall z_2 \dots \exists y_{k-1} \forall z_{k-1} \exists y_k \,\, \psi(\overline{y}, \overline{z})\]
\stfa  $\calA \in \struc[\sigma]$,
$f(\enc_\sigma(\calA))$ is equal to the number of tuples $(f_1,\dots,f_k)$ of functions such that 
$$\mathcal{A} \ \vDash \ \forall z_1 \dots \forall z_{k-1} \ \psi(f_1, f_2(z_1), \dots, f_k(z_1, \dots, z_{k-1}), z_1, \dots, z_{k-1})$$
\end{definition}

This means that $\cPiSk$ contains those functions that, for a fixed \FO-formula over some vocabulary $\sigma$, map an input $w$ to the number of Skolem functions on $\mathcal{A} = \enc_\sigma^{-1}(w)$.


\begin{theorem}\label{thmNumACeqWinFO}
$\FO\textrm{-uniform } \cAC = \cPiSk$
\end{theorem}

The above mentioned result $\FO=\AC$ \cite{baimst90,ImmermanBuch} requires built-in order and BIT; hence it is no surprise that also for the just given theorem these relations are needed, and this is the reason why they also appear in Def.~\ref{def-pisk}.


\section{Connecting the Characterizations of \cAC and \cP}
\label{sect-framework}

We will now establish a unified view on the model-theoretic characterizations of both \cAC and \cP. This will be done by viewing \cAC as a syntactic subclass of \cFO. 
In Theorem \ref{thmNumACeqWinFO} we characterized \cAC by a process of counting assignments to function variables in FO-formulae, but only in a very restricted setting. 
It is natural to define the process of counting functions in a more general way, similar to the framework of \cite{DescCompNumP}, repeated here in Def.~\ref{def-forel}, where Saluja et al.{} count assignments to free relation variables in FO-formulae to obtain their characterization of \cP.

\begin{definition}
$\cFO$ is the class of all functions $f\colon \{0,1\}^* \rightarrow \mathbb{N}$ for which there is a vocabulary $\sigma$, including built-in $\leq, \textrm{BIT}$ and $\min$, and an \FO-formula $\varphi(F_1, \dots, F_k, x_1, \dots, x_\ell)$ over $\sigma$ with free function variables $F_1, \dots, F_k$ and free individual variables $x_1, \dots, x_\ell$ \stfa $\calA \in \struc[\sigma]$,
\[f(\enc_\sigma(\calA)) = \bigl|\bigl\{(f_1, \dots, f_k, c_1, \dots, c_\ell) \;\big|\; \calA \vDash \varphi(f_1, \dots, f_k, c_1, \dots, c_\ell\bigr\}\bigr|.\]
\end{definition}

In the same fashion we define counting classes using fragments of \FO, such as $\cSigmai{k}$ and $\cPii{k}$ for arbitrary $k$. Note, that the free individual variables could also be seen as free function variables of arity 0. 

We stress that our signatures in the above definition include symbols $\leq$, $\textrm{BIT}$, and $\min$ with their standard interpretations; as argued already several times, these built-ins are necessary in order to obtain a close correspondence between Boolean circuits and first-order logic. In contrast to our definition, Saluja et al.{} (see Def.~\ref{def-forel}) only use built-in order. Still, we will now see that both concepts, counting relations and counting function, are in fact equivalent as long as we use all of \FO, even with different sets of built-in relations.

\begin{theorem}
$\cFO^\rel = \cFO = \cP$.
\end{theorem}

\begin{proof}
The inclusion $\cFO^\rel \subseteq \cFO$ is shown as follows:
Let $f \in \cFO^\rel$ via the formula $\varphi$ containing  free relation variables $R_1, \dots, R_k$. We can replace $R_i$ by a function variable $F_i$ of the same arity for all $i$. We then add a conjunct to the formula ensuring that for these functions only $\min$ and the element $x > \min$ with $\forall y (y < x \rightarrow y = \min)$ are allowed as function values. Then each occurrence of $R_i(\overline{z})$ can be replaced by $F_i(\overline{z}) = \min$.

The inclusion $\cFO \subseteq \cP$ is trivial. The inclusion $\cP\subseteq\cFO^\rel$ was shown in \cite{DescCompNumP}.
\end{proof}



Note that $\cAC=\cPiSk$ 
does not directly arise from this definition by choosing an appropriate fragment of \FO because of the restricted usage of the second-order variables in Def.~\ref{def-pisk}.
Still we will characterize \cAC as a syntactic subclass of \cFO as follows.

\begin{definition}
Let \cPiPrefix be the class of all functions $g$ that can be characterized as
\[g(w) = \bigl|\bigl\{(\tu{f}, \tu{c}) \;\big|\; \varphi(\tu{f}, \tu{c}))\bigr\}\bigr|,\]
where $\varphi(\tu{F}, \tu{x}) = \forall y_1 \dots \forall y_k \psi(\tu{F}, \tu{x}, y_1, \dots, y_k)$ is a $\Pi_1$ formula and all arity-$a$ functions (for any $a$) occur in $\psi$ only in the form $F(y_1, \dots, y_a)$.
\end{definition}

\begin{lemma}
$\FO\textrm{-uniform } \cAC = \cPiPrefix$
\end{lemma}

\begin{proof}
By Theorem~\ref{thmNumACeqWinFO} it suffices to show 
$ \cPiSk \mathop{=} \cPiPrefix$.
We consider first  the inclusion $\cPiSk \subseteq \cPiPrefix$. Let $g \in \cPiSk$ via a formula $\varphi$  as in Definition \ref{def-pisk}. Then we can simply replace the occurrences of variables $y_i$ in $\psi$ by the corresponding function terms. The resulting formula is prefix-restricted as needed and directly shows $g \in \cPiPrefix$.

For $\cPiPrefix \subseteq \cPiSk$, let $g \in \cPiPrefix$ via a formula $\varphi$. Since all function symbols occurring in $\varphi$ are only applied to a unique prefix of the universally quantified variables, they can be seen as Skolem functions of suitable existentially quantified variables. Thus, we can replace the occurrences of the function symbols by new variables that are existentially quantified at adequate positions between the universally quantified variables. If for example, the input for a function was $x_1, \dots, x_\ell$, then the new  variable is quantified after the part
$\forall x_1 \dots \forall x_\ell$
of the quantifier prefix. This yields a formula $\varphi'$ that shows $g \in \cPiSk$.
\end{proof}

\section{An Alternation Hierarchy in $\cFO$}
\label{sect-hierarchy}

In this section we study a hierarchy within \cFO based on quantifier alternations. Interestingly, our approach allows us to locate \cAC in this hierarchy. First we note that the whole hierarchy collapses to a quite low class.

\begin{theorem}
$\cFO = \cPii{1}$
\end{theorem}

\begin{proof}
Let $f \in \cFO$ via a $\FO$-formula $\varphi(\tu{f}, \tu{x})$ in prenex normal form. 
We show how to transform $\varphi$ to a $\Pii{1}$-formula also defining $f$.
As a first step, we change $\varphi$ in such a way that for each existential variable instead of ``there is an $x$'' we say ``there is a smallest $x$''. Formally, this can be done with the following transformation:
\[\exists y \theta(y) \leadsto \exists y (\theta(y) \land \forall z (\neg \theta(z) \lor y < z \lor y = z))\]
applied recursively to all  existential quantifiers in $\varphi$.  Note that now for every satisfied $\exists$-quantifier there is exactly one witness. 

For the sake of argument,  suppose that after the above transformation and re-conversion  to prenex normal form the formula $\varphi(\tu{f}, \tu{x})$ corresponds to $\varphi'(\tu{f}, \tu{x})$, where 
\[\varphi'(\tu{f}, \tu{x}) = \exists z_1 \forall y_1 \dots \forall y_{\ell-1} \exists z_\ell \psi(\tu{f}, \tu{x},z_1, \dots, z_\ell, y_1, \dots, y_{\ell-1}).\]
Looking at the Skolemization of $\varphi'$, our transformation ensures that every existentially quantified variable  has a unique Skolem function. Thus,
\[\varphi''(\tu{f}, g_1, \dots, g_\ell) = \forall y_1 \dots \forall y_{\ell-1} \psi(\tu{f}, g_1, g_2(y_1), \dots, g_\ell(y_1, \dots, y_{\ell-1}), y_1, \dots, y_{\ell-1})\]
shows $f \in \cPii{1}$.

\end{proof}

Next we look at the lowest class in our hierarchy and separate it from \cAC.

\begin{theorem}\label{Sigma0_in_cAC}
$\cSigmai{0} \subsetneq \FO\textrm{-uniform }\cAC$
\end{theorem}

\begin{proof}
We start by showing the inclusion. Certain observations in that proof will then almost directly yield the strictness. Let $f \in \cSigmai{0}$ via the quantifier-free FO-formula $\varphi(x_1, \dots, x_k, F_1, \dots, F_\ell)$ over some vocabulary $\sigma$, where $x_1, \dots, x_k$ are free individual variables and $F_1, \dots, F_\ell$ are free function variables, that is,
\[f(\enc_\sigma(\calA)) = |\{(c_1, \dots, c_k, f_1, \dots, f_\ell) \mid \calA \vDash \varphi(c_1, \dots, c_k, f_1, \dots, f_\ell)\}|.\]
\WLOG we can assume that in $\varphi$ no nesting of functions occurs. If there is an occurrence of a function $G$ as an argument for function $H$, then we can replace the occurrence of $G$ by a new free variable and force this variable to be equal to the function value. This ensures that there is only one unique assignment to this new free variable.

Let $A \dfn \textrm{dom}(\calA)$. For all $i$, let $a_i$ be the arity of $F_i$ and let $m_i$ be the number of syntactically different inputs to occurrences of $F_i$ within $\varphi$. Furthermore, let $e_{i1}, \dots, e_{im_i}$ be those inputs in the order of their occurrence within $\varphi$ and let $\varphi'(x_1, \dots, x_k, y_{11}, \dots, y_{1m_1}, \dots, y_{\ell 1}, \dots, y_{\ell m_\ell})$ be $\varphi$ after replacing for all $i,j$ all occurrences of $F_i(e_{ij})$ by the new free variable $y_{ij}$. Let $m \dfn \sum_i{m_i}$.

Considering a fixed assignment to the variables $x_1, \dots, x_k$, the idea now is to use free individual variables in order to count the number of assignments to the terms $f_i(e_{ij})$ for all $(i,j)$. After that, all $f_i$ have to be chosen in accordance with those choices to get the correct number of functions that satisfy the formula. Formally, this is done as follows:
\begin{align*}
f(\enc_\sigma(\calA)) & = \sum_{\tu{c} \in A^k}
\sum_{\substack{(f_1, \dots, f_\ell) \in\\A^{A^{a_1}} \times \dots \times A^{A^{a_\ell}}}}
[\calA \vDash \varphi(c_1, \dots, c_k, f_1, \dots, f_\ell)]\\
& = \sum_{\tu{c} \in A^k} \quad \sum_{\tu{d} \in A^m} \sum_{(f_1, \dots, f_\ell) \in G} [\calA \vDash \varphi'(\tu{c}, \tu{d})],
\end{align*}
where $G \dfn \{\tu{f} \in A^{A^{a_1}} \times \dots \times A^{A^{a_\ell}} \mid \forall (i,j):\ \calA \vDash \ d_{ij} = f_i(e_{ij})\}$.

Since $[\calA \vDash \varphi'(\tu{c}, \tu{d})]$ does not depend on $(f_1, \dots, f_\ell)$, we can multiply by the cardinality of $G$ instead of summing:
\[
f(\enc_\sigma(\calA)) = \sum_{\substack{\tu{c} \in A^k,\\\tu{d} \in A^m}} [\calA \vDash \varphi'(\tu{c}, \tu{d})] \cdot |G|
\]
Now we are in a position to show $f \in \FO\textrm{-uniform } \cAC$.

The sum only has polynomially many summands and thus is obviously possible in \FO-uniform \cAC.

For $[\calA \vDash \varphi'(\tu{c}, \tu{d})]$, the circuit only has to evaluate a quantifier-free formula depending on an assignment that is given by the path from the root to the current gate. This is similar to the corresponding part of the proof of $\FO = \FO\textrm{-uniform } \cAC$ and thus can be done in $\AC \subseteq \cAC$.

For $|G|$ we first note that the total number of possible assignments for $\tu{f}$ is
\[|A^{A^{a_1}} \times \dots \times A^{A^{a_\ell}}| = |A|^{\sum_i{|A|^{a_i}}}.\]
In $G$, the choices for the variables $d_{ij}$ fix for all $i$ up to $m_i$ function values of $f_i$. This means, that at least $|A|^{a_i}-m_i$ choices of function values can be arbitrarily chosen. 

If for some $(i,j)$, $e_{ij}$ is semantically equal to $e_{ij'}$ for some $j' < j$, it has to hold that $d_{ij} = d_{ij'}$. Additionally, this reduces the amount of function values that are fixed by the $d_{ij}$ by 1. To make this formal we define for any $(i,j)$
\[S_{ij} = \{j' \mid j' < j \textrm{ and } \calA \vDash \ e_{ij} = e_{ij'}\}.\]
From the above considerations we get
\[
|G| = [\bigwedge_{(i,j)} \bigwedge_{j'} (j' \in S_{ij}) \to d_{ij} = d_{ij'}] \cdot |A|^{\sum_i{|A|^{a_i}} - \sum_i{m_i}} \cdot |A|^{\sum_{ij}{[S_{ij} \neq \emptyset]}}.
\]

Since the $a_i$ and $m_i$ are constants and $S_{ij}$ is \FO-definable, $|G|$ can be computed in \FO-uniform \cAC. This concludes the proof for $\cSigmai{0} \subseteq \FO\textrm{-uniform } \cAC$.

Now note that the proof above also shows that for any $\cSigmai{0}$-function $f$, either for all inputs $w$, $f(w)$ is polynomially bounded in $|w|$ or for all inputs $w$, $f(w)$ is always divisible by $|w|^{\sum_i{|w|^{c_i}} - \textrm{ konst}}$ for constants $c_i > 0$. Thus, the function $f(w) = |w|^{\lceil|w|/2\rceil} \in \FO\textrm{-uniform } \cAC$ is not in $\cSigmai{0}$ which means $\cSigmai{0} \neq \FO\textrm{-uniform } \cAC$.
\end{proof}


\begin{theorem}\label{thm-numAC0 is not numP}
$\cPiSk \subsetneq \cPii{1}$.
\end{theorem}

\begin{proof}
From the above we know that the left class is equal to \FO-uniform $\cAC$ and the right class is equal to $\cP$. Strict inclusion now follows immediately from the following considerations:
Let $\mathcal{F}$ be a class of functions $\{0,1\}^* \rightarrow \mathbb{N}$. Then the class $\mathsf{C} \cdot \mathcal{F}$ is the class of all of languages $A$ for which there are $f,g \in \mathcal{F}$ \stfa $x \in \{0,1\}^*$,
$x \in A \Leftrightarrow f(x) > g(x)$.
In \cite{PAC0TC0} it was shown that $\TC = \mathsf{C} \cdot \cAC$. Also, it is well known that $\textrm{PP} = \mathsf{C} \cdot \cP$. Allender's separation $\TC \neq \textrm{PP}$ \cite{Allender99} now directly yields \FO-uniform $\cAC \neq \cP$.
\end{proof}

So far we have identified the following hierarchy:
\begin{equation}\label{first-inclusion-chain}
\cSigmai{0} \subsetneq \cPiSk = \cAC \subsetneq \cPii{1} = \cP.	
\end{equation}

Next we turn to the class \cSigmai{1} and show that it forms a different branch between $\cSigmai{0}$ and $\cPii{1}$.

\begin{lemma}\label{l1}
There exists a function $F$  which is in $\cPiSk$ but not in $\cSigmai{1}$.
\end{lemma}

\begin{proof}Let $\tau=\{E,c,d, \leq,\textrm{BIT},\min\}$ where $E$ is a binary relation symbol and $c,d$ are constant symbols. Let us consider  the function $F$ defined by the number of Skolem functions of variable $z$ in the formula $\varphi = \forall x \forall y \exists z \ \psi(x,y,z)$ with
\[
\psi = (E(x,y)\rightarrow z=c \vee z=d) \wedge (\neg E(x,y) \rightarrow z=c).
\]
For a given $\tau$-structure $\calA$ with $c^\calA \neq d^\calA$, it is clear that:
\[
F(\enc_\tau(\calA)) = |\{ f \mid \calA \models   \forall x \forall y \ \psi(x,y,f(x,y))\}| = 2^{ |E^{\calA}|},
\]
since each edge gives rise to two possible values for $z=f(x,y)$ and  each non edge to only one value. Thus, $F \in \cPiSk$.

Suppose now that $F\in \cSigmai{1}$ i.e. that  there exists $\phi(\tu x,\tu g)\in \Sigmai{1}$ such that for all $\tau$-structures $\mathcal{G}$,
\[
F(\enc_\tau(\mathcal{G})) = |\{(\tu a,\tu g_0) \mid \mathcal{G} \models \phi(\tu a,\tu{g}_0)\}|
\]
and in particular for $\calA$ as above,
\[2^{|E^\calA|} = F(\enc_\tau(\calA)) = |\{(\tu a,\tu g_0) \mid \calA \models \phi(\tu a,\tu{g}_0)\}|.\]
Now consider the following structure $\calA'$ defined simply by  extending $\textrm{dom}(\calA)=\{0,...,n-1\}$ by two new elements, i.e., $ \textrm{dom}(\calA')=\{0,...,n+1\}$. Note that  $E^{\calA}=E^{\calA'}$, hence the two structures have the same number of edges. To make the presentation simpler, suppose $\tu g=g$ and that the arity of $g$ is one. Any given $g_0:\textrm{dom}(\calA)\longrightarrow \textrm{dom}(\calA)$, can be extended in several ways on the domain $\textrm{dom}(\calA')$ in particular as $g_1$ and $g_2$ below:
 \begin{itemize}
 \item $g_1(x)=g_0(x)$ for all $x\in \textrm{dom}(\calA)$ and $g_1(n)=c$, $g_1(n+1)=d$.
  \item $g_2(x)=g_0(x)$ for all $x\in \textrm{dom}(\calA)$ and $g_2(n)=d$, $g_2(n+1)=c$.
 \end{itemize} 
 
 Formulas in $\Sigmai{1}$ are stable under extension of models so if  $\tu a$ and $\tu g_0$ are such that $\calA \models \phi(\tu a,\tu{g}_0)$ then $\calA' \models \phi(\tu a,\tu{g}_1)$ and $\calA' \models \phi(\tu a,\tu{g}_2)$. Hence,
 \[
 |\{(\tu a,\tu g') \mid \calA' \models \phi(\tu a,\tu{g}')\}| > |\{(\tu a,\tu g_0): (\calA,\tu a,\tu g_0)\models \phi(\tu x,\tu g)\}|. \] 
On the other hand,  $F(\enc_\tau(\calA)) = F(\enc_\tau(\calA'))$ holds, hence our assumpion that  $\phi(\tu x,\tu g)\in \Sigmai{1}$ defines $F$ has led to a contradiction.
\end{proof}

For the opposite direction, we first show the following lemma.

\begin{lemma}\label{lem-cDnf-numP-complete}
The function $\cDnf$ is complete for \cP under \AC-Turing-reductions.
\end{lemma}

\begin{proof}
It is known that $\#3\mathrm{DNF}$ is $\cP$-complete under metric reductions. A metric reduction of the $\cP$-complete problem $\#3\mathrm{CNF}$ to $\#3\mathrm{DNF}$ is as follows: Given a 3CNF-formula $\varphi$ over $n$ variables, we first construct $\varphi'=\neg\varphi$. This is a 3DNF-formula. Obviously, the number of satisfying assignments of $\varphi$ is equal to $2^n$ minus the number of satisfying assignments of $\varphi'$. Since this reduction can be computed by an \AC-circuit and moreover $\#3\mathrm{CNF}$ is $\cP$-complete under \AC-reductions (as follows from the standard proof of the NP-completeness of SAT), $\#3\mathrm{DNF}$ is complete for $\cP$ under $\AC$-Turing-reductions.
\end{proof}

\begin{lemma}\label{lem-Sigma1-notin-PiSk}\label{l2}
There exists a function $F$ which is in $\cSigmai{1}$ but not in $\cPiSk$.
\end{lemma}

\begin{proof}
First note that $\textrm{FTC}^0 \neq \cP$: Making use of the complexity-theoretic operator $\mathsf C$ (see proof of Theorem~\ref{thm-numAC0 is not numP}), we obtain
$\PP = \mathsf{C}\cdot \cP  \subseteq \mathsf{C}\cdot \mathrm{FTC}^0 = \TC$, but $\TC \neq \PP$ \cite{all96}.

We now show this lemma by modifying the counting problem \cDnf to get a \cP-complete function inside of \cSigmai{1}. If the reduction we use can be computed in $\textrm{FTC}^0$, the modified version of \cDnf can not be in $\cPiSk \subseteq \textrm{FTC}^0$, because this would contradict $\textrm{FTC}^0 \neq \cP$.

Consider the vocabulary $\sigma_\dnf$ and the formula $\Phi_\cDnf$ from example \ref{example-3DNF}. Let $\sigma$ be the vocabulary extending $\sigma_\dnf$ with built-in $\leq$, BIT and min. To get a function in \cSigmai{1}, we need to use a free function variable instead of the free relation variable $T$. This will surely blow up the function value. The idea is to make sure that compared to \cDnf, the function value only changes by a factor depending on the input length, not on the specific satisfying assignments. To achieve this, we interpret all even function values as false and all odd function values as true. Thus, the number of 1's and 0's in a satisfying assignment do not influence the blowup.

Following this idea we define for all $\sigma$-structures $\calA$
\[\cDnfF(\enc_{\sigma}(\calA)) = |\{(f) \mid \calA \vDash \Phi_\cDnfF(f)\}|,\]
where $\Phi_\cDnfF$ is $\Phi_\cDnf$ after replacing for all variables $x$ subformulae of the form $T(x)$ by $\textrm{BIT}(\min, f(x))$. By definition, $\cDnfF \in \cSigmai{1}$.

We now want to reduce \cDnf to \cDnfF. Since the idea with the blowup only depending on the input size only works, if the universe has even cardinality, the first step of the reduction is doubling the size of the universe. Let $\calA$ be a structure and $\calA'$ the structure that arises from $\calA$ by doubling the size of the universe. 
Let $A=\{0,\dots, n-1\}$ and $A'=\{0,\dots,2n-1\}$ be their respective universes. Each assignment for $T$ with $\calA \vDash \Phi_\cDnf(T)$ gives rise to the following set of assignments for $f$ with $\calA' \vDash \Phi_\cDnfF(f)$:
\[S_T = \{f:A' \rightarrow A' \mid \textrm{for all } x \in A: f(x) \equiv 1 \mod 2 \quad \Leftrightarrow \quad T(x)\}.\]

These sets are disjunct and by definition of $\Phi_\cDnfF(f)$ their union is equal to $\{f \mid \calA' \vDash \Phi_\cDnfF(f)\}$. For each $T$, the functions $f$ in $S_T$ have $n$ choices for $f(x)$, if $x \in A$ and $2n$ choices, if $x \notin A$. Thus, $|S_T| = |A|^{|A|} \cdot (2\cdot |A|)^{|A|}$, yielding
\[\cDnf(\enc_{\sigma_\dnf}(\calA)) = \frac{\cDnfF(\enc_{\sigma}(\calA'))}{|A|^{2|A|} \cdot 2^{|A|}}.\]

Doubling the size of the universe can be done in \FO-uniform $\textrm{FTC}^0$ by adding the adequate number of 0-entries in the encodings of all relations.

The term $|A|^{2|A|} \cdot 2^{|A|}$ can be computed in \FO-uniform $\cAC \subseteq \FO\textrm{-uniform } \textrm{FTC}^0$ and division can be done in $\FO\textrm{-uniform } \textrm{FTC}^0$ due to \cite{HesseDivTC0}.

Since \cDnf is \cP-complete under \AC-Turing-reductions by Lemma \ref{lem-cDnf-numP-complete}, this means that \cDnfF is \cP-complete under \TC-Turing-reductions.
\end{proof}

So Lemmas \ref{l1} and \ref{l2} show that $\cSigmai{1}$ and $\cPiSk$ are incomparable, and we obtain the inclusion chain 
$\cSigmai{0} \subsetneq \cSigmai{1} \subsetneq \cPii{1} = \cP$.
Together with (\ref{first-inclusion-chain}) we therefore obtain
\[\cSigmai{0} 
\begin{array}{@{\ }r@{\ }c@{\ }l@{\ }}
\rotatebox[origin=c]{30}{$\subsetneq$} 
& \raisebox{.75ex}{\cAC = \cPiPrefix} 
& \rotatebox[origin=c]{-30}{$\subsetneq$} 	\\
\rotatebox[origin=b]{-20}{$\subsetneq$} 
& \raisebox{-.75ex}{\cSigmai{1}} 
& \rotatebox[origin=c]{20}{$\subsetneq$}	
\end{array}
\cPii{1} = \cFO = \cP
\tag{\ref{eqInclusions}}
\]

\section{Hierarchy Based on the Number of Universal Variables}
\label{sect-varhierarchy}

In this section we study another hierarchy in \cFO based on syntactict restrictions, this time given by the number of universal variables.

Let $\cPiik{1}{k}$ denote the class of $\Pi_1$ formulae of the form
\[\forall x_1\cdots \forall x_m \psi,  \]
where $m\le k$, and $\psi$ is a quantifier-free formula. We will show that 
\begin{equation}\label{var_hierarchy}
  \cPiik{1}{k} \subsetneq \cPiik{1}{k+1}, 
\end{equation}
for all $k\ge 1$. These results can be shown by applying  a result of Grandjean and Olive which we will discuss next.
\begin{definition}
 We denote by   $\ESOfvar{k}$  the class of   ESO-sentences  in Skolem normal form
\begin{equation*}
\exists f_1\ldots\exists f_n \forall x_1\ldots \forall x_r \psi,
\end{equation*}
where $r\le k$, and $\psi$ is a quantifier-free formula.
\end{definition}

 It was shown in  \cite{grandjean04}  that with respect to any finite signature $\sigma$
\[ \ESOfvar{k}=\NTIME_{\RAM}(n^k), \]
where $\NTIME_{\RAM}(n^k)$ denotes the family  of classes of $\sigma$-structures that can be recognized by a non-deterministic RAM in time $O(n^k)$. Note that by  \cite{cook72},
\[ \NTIME_{\RAM}(n^k)\subsetneq \NTIME_{\RAM}(n^{k+1}).\]
These results can be used to show the strictness of the variables hierarchy (see \eqref{var_hierarchy}). For the case $k=1$ of Theorem \ref{hierarchy_k} we use  the following lemma which holds for vocabularies equipped with built-in order  $<$ and constants $min$ and $max$. 
\begin{lemma}\label{suc_lemma}
Let $\sigma$ be a vocabulary including built-in $<$, $min$ and $max$. Then there is a formula $Succ\in \Pi^{1}_1$ with free unary function variables $s$ and $p$ such that for all $\calA$, ${\bf s}$, and  ${\bf p}$,  $\calA\models Succ({\bf s},{\bf p})$ iff
\begin{enumerate}
\item  for all $e<max^{\calA}$: ${\bf s}(e)=e+1$, and  ${\bf s}(max^{\calA})=max^{\calA}$,
\item  for all $e>min^{\calA}$: ${\bf p}(e)=e-1$, and  ${\bf p}(min^{\calA})=min^{\calA}$.
\end{enumerate}
\end{lemma}


\begin{proof}
It is straighforward to check that $Succ$ can be defined by  universally quantifying $x$ over the conjunction of the following clauses: 
\begin{itemize}
\item $(x<max \wedge min<x)\rightarrow \big((x<s(x)\wedge p(x)<x)\wedge(p(s(x)) = s(p(x)) = x)\big)$,
\item $p(min)=min \wedge s(max)=max$.
\end{itemize}
\end{proof}

\begin{theorem}\label{hierarchy_k} Let $k\ge 1$. Then (assuming the auxiliary built-in constant $max$ for the case $k=1$)
$$ \cPiik{1}{k}  \subsetneq \cPiik{1}{k+1}.$$ 
\end{theorem}
\begin{proof} We consider first the case $k=1$. Let us fix $\sigma=\{<,\textrm{BIT},min,max,P\}$, where $P$ is unary. By the above there exists  a sentence $\exists f_1\cdots \exists f_n\psi \in \ESOfvar{k+1}[\sigma]$ defining a binary language $L$ which cannot be defined by any sentence $\chi \in \ESOfvar{k}[\sigma]$. We claim that the function $F$ associated with the formula $\psi(f_1,...,f_n)\in \Pi^{k+1}_1$,
\[ F(\enc_\sigma(\calA)) =  |\{ (f_1,...,f_n): \calA\models  \psi(f_1,...,f_n)\}|,
\]
 is not a member of  $\cPiik{1}{k}$. For a contradiction, assume that $F\in \cPiik{1}{k}$. Then there exists a formula $\chi(\tu y, \tu g)\in \Pi^{k}_1$ such that  
\[ F(\enc_\sigma(\calA)) =  |\{(\tu y, \tu g): \calA\models \chi(\tu y, \tu g) \}|
\]
By the above,  the sentence $\exists \tu g \exists \tu y \chi$ defines the language $L$. The variables $\tu y=y_1,\ldots,y_r$ can be replaced in $\chi$ by fresh unary function symbols $ g_{y_1},...,g_{y_r}$ whose interpretations are forced to be  constant functions in the following way, $x_1$ being the universal variable in $\chi$.
\begin{enumerate}
\item We replace all occurrences of $y_i$ in $\chi$ by the term  $g_{y_i}(x_1)$,
\item  We add to the quantifier-free part of $\chi$ a  conjunct $g_{y_i}(x_1)=g_{y_i}(s(x_1))$, 
for $1 \le i \le r$,
\item We add to the quantifier-free part of $\chi$ the conjunt  $Succ(x_1)$ defined in Lemma \ref{suc_lemma} which forces the interpretation of the unary function $s$ to be the unique successor function associated to $<$. 
\end{enumerate}
 Now   $\exists \tu g \exists \tu g_y \exists s\exists p \chi \in \ESOfvar{k}$,  and it defines the language $L$. But this contradicts the assumption that $L$ cannot be defined any sentence of $\ESOfvar{k}$.

Let us then consider the case $k\ge 2$. Now $\sigma=\{<,\textrm{BIT},min,P\}$, and we proceed analogously to the case $k=1$ except that  the
 formulae of items 2. and 3. are replaced by formulae $g_{y_i}(x_1)=g_{y_i}(x_2)$, where $x_1, x_2$ are two different universal variables in $\chi$.
 \end{proof}

It is an interesting open question to study the relationship of  $\cAC$ with the classes $\cPiik{1}{k}$.

\section{$\cAC$ compared to the classes from Saluja et al.}
\label{sect-connection}

In this section we study the relationship of $\cAC$ to the syntactic classes introduced in \cite{DescCompNumP}. As in \cite{DescCompNumP},
 these classes are defined assuming a built-in order relation only.

\begin{theorem} 
\begin{itemize}
\item $\cSigmaiRel{0}\subsetneq \cAC$,
\item Let $\mathcal{C}\in \{\cSigmaiRel{1},  \cPiiRel{1}, \cSigmaiRel{2} \}$. 
Then the following holds: $\cAC \not \subseteq \mathcal{C}$ and   $ \mathcal{C}\not \subseteq \cAC $.
\end{itemize}
\end{theorem} 
\begin{proof}
The proof of the inclusion $\cSigmaiRel{0}\subsetneq \cAC$ is analogous to the proof of Theorem \ref{Sigma0_in_cAC} and is thus omitted. 

Next, the claim $\mathcal{C}\not \subseteq \cAC$ for $\mathcal{C}\in \{\cSigmaiRel{1},  \cPiiRel{1}, \cSigmaiRel{2} \}$ can be proven as follows: From Example~\ref{example-3DNF} we know that $\#3\mathrm{DNF}\in \mathcal{C}$ and from Lemma \ref{lem-cDnf-numP-complete} we know that \cDnf is \cP-complete under \AC-Turing-reductions. Now suppose $\#3\mathrm{DNF} \in \cAC$. Then $\cP \subseteq {\mathrm{FAC}^0}^{\cAC} \subseteq \mathrm{FTC}^0$ \cite{HaVo16}, contradicting $\mathrm{FTC}^0 \neq \cP$, which was shown in the proof of Lemma \ref{lem-Sigma1-notin-PiSk}. Hence $\#3\mathrm{DNF} \not \in  \cAC$ and $\mathcal{C}\not \subseteq \cAC$.

It remains to show $\cAC \not \subseteq \mathcal{C}$. We show this by an argument similar to the proof that $\#$HA\-MIL\-TON\-IAN is not in $\cSigmaiRel{2}$,  showing the separation of  \cSigmaiRel{2} from $\cFO$, see Theorem 2  in \cite{DescCompNumP}. We will show that a very simple function $f$ on encodings of $\tString$-structures is not in $\mathcal{C}$. Define $f$ as follows: $f(w) = 1$, if $|w|$ is even, and  $f(w) = 0$ otherwise. Obviously $f\in  \cAC$. It now suffices to show that $f\not \in  \cSigmaiRel{2}$. For contradiction, assume that $f \in  \cSigmaiRel{2}$ via a formula $\phi(\tu x,\tu R)\in \Sigma_2^\rel$, where
 \[  \phi(\tu x,\tu R)= \exists \tu u \forall \tu v \theta (\tu u, \tu v, \tu x, \tu R),    \]
and $\theta$ is a quantifier-free formula. Let $s$ and $t$ be the lengths of the tuples $\tu u $ and $\tu x$, respectively. Let $n\ge s+t+1$ be even and let $w \in \{0,1\}^n$. By the assumption, there exists $\tu {\bf u}$,  $\tu {\bf x}$, $\tu {\bf R}$ such that 
\[  \calA_w \models \forall  \tu v \theta (\tu {\bf u}, \tu {\bf x}, \tu {\bf R}). \]
By the choice of $n$, we can find $i\in \{0,..,n-1\}$ such that $i$ does not appear in the tuples $\tu {\bf u}$ and $\tu {\bf x}$. Let $\calA_{w'}$
denote the structure arising by removing the element $i$ from the structure  $\calA_{w}$, and let  $\tu {\bf R}^*$ denote the relations of 
 $\calA_{w'}$ arising by removing tuples with the element $i$ from  $\tu {\bf R}$. By closure under substructures of universal first-order formulae, it follows that 
\[  \calA_{w'} \models \forall  \tu v \theta (\tu {\bf u}, \tu {\bf x}, \tu {\bf R}^*), \]
implying that  $f(\calA_{w'})\ge 1$.  But  $|w'|$ is odd and hence  $f(\calA_{w'})=0$ contradicting the assumption that  the formula $\phi(\tu x,\tu R)$ defined the function $f$.
\end{proof}


Last, we want to give another inclusion result between one of our classes and a class from the Saluja et al.{} hierarchy.

\begin{lemma}
There exists a function $F$  which is in $\cSigmaiRel{1}$ but not in $\cSigmai{1}$.
\end{lemma}

\begin{proof} We prove that $\#3\mathrm{DNF}$ is not in  $\cSigmai{1}$ though it belongs to $\cSigmaiRel{1}$.  As in Example~\ref{example-3DNF}, we use the vocabulary $\sigma_{\#3\mathrm{DNF}}=(D_0,D_1,D_2,D_3)$ and consider the vocabulary $\sigma$ extending $\sigma_{\#3\mathrm{DNF}}$ with  built-in linear order $\leq$, BIT and min. 
Suppose $\#3\mathrm{DNF}$ is definable by a $\sigma$-formula $\Phi(\tu x,\tu g)\in\Sigmai{1}$. To a given $\mathrm{DNF}$ formula, $\varphi$, with $n\geq 2$ variables,  one associates a $\sigma$-structure $\calA_{\varphi}$ such that the number $m$ of satisfying assignments of $\varphi$ is equal to
\[
m = |\{(\tu a,\tu g_0) \mid \calA_{\varphi} \models \Phi(\tu a,\tu{g}_0)\}|
\]
Let $\{0,...,n-1\}$ be the domain of  $\calA_{\varphi}$. Consider the structure $\calB$ extending  $\calA_{\varphi}$ with one additional element $n$, correctly extending the numerical predicates. Structure $\calB$ encodes a formula $\varphi'$ whose number of satisfying assignments is obviously $2m$. 
Formulas in $\Sigmai{1}$ are stable by extension, so for any fixed $(\tu a,\tu g_0)$ such that $\calA \models \Phi(\tu{a}, \tu{g}_0)$, any extension $\tu g_0'$ of $\tu g_0$ on the domain $\{0,...,n\}$ of $\calB$ is such that  $\calB \models \Phi(\tu a, \tu{g}_0')$. 
Each $g\in \tu g_0$ of arity $a\geq 1$ defined on $\{0,...,n-1\}$ can be extended on  $\{0,...,n\}$ in at least $(n+1)^{\sum_{i=1}^{a} \binom{a}{i}n^{a-i}}\geq n+1$ ways. Hence:
\[
|\{(\tu a,\tu g_0') \mid \calB \models \Phi(\tu a, \tu{g}_0')\}| \geq (n+1)m > 2m
\]
\noindent contradicting the assumption that $\Phi(\tu x,\tu g)\in\Sigmai{1}$ defines $\#3\mathrm{DNF}$.
\end{proof}

\section{Conclusion}
\label{sect-conclusion}

In this paper we have started a descriptive complexity approach to arithmetic computations. We have introduced a new framework to define arithmetic functions by counting assignments to free function variables for first-order formulae. Compared to a similar definition of Saluja et al.{} where assignments to free relational variables are counted, we obtain a hierarchy with a completely different structure, different properties and different problems. The main interest in our hierarchy is that it allows the classification of arithmetic circuit classes such as $\cAC$, in contrast to the one from Saluja et al. 

We have only started the investigation of our framework, and many questions remain open for future research: 

\begin{asparaenum}

\item 
Sipser proved a depth hierarchy within the Boolean class $\AC$ \cite{sip83b}. This hierarchy can be transfered into the context of arithmetic circuits: There is an infinite depth hierarchy within \cAC. Does this circuit hierarchy lead to a logical hierarchy within $\cPiSk$? Maybe it is possible to obtain a hierarchy defined by limiting the arity of the Skolem functions.

\item
The connection between \cAC and the variable hierarchy studied in Sect.\ref{sect-varhierarchy} is not clear. We think it would be interesting to study if \cAC is fully contained in some finite level of this hierarchy. 

\item
One of the main goals of Saluja et al.{} in their paper \cite{DescCompNumP} was to identify feasible subclasses of $\cP$. They showed that $\cSigmaiRel{0}$-functions can be computed in polynomial time, but even more interestingly, that functions from a certain higher class $\#R\Sigma_2$ allow a full polynomial-time randomized approximation scheme. Are there approximation algorithms or even schemes, maybe randomized, for some of the classes of our hierarchy?

\item
The most prominent small arithmetic circuit class besides $\cAC$ is problably the class $\cNC$ \cite{camcthvo96}. Can it be characterized in our framework or by a natural extension of it, for example by allowing generalized quantifiers? The Boolean class $\NC$ is obtained by first-order formulae with Lindström quantifiers for group word problems; i.e., we have, very informally, that $\AC=\FO$ and $\NC=\FO+\text{GROUP}$, see \cite{baimst90,HeribertBuch}.

\item
In Sect.~\ref{sect-connection}, we clarified the inclusion relation between the class \cAC and all classes of the Saluja et al.{} hierarchy, and we gave a small number of examples for (non-)inclusion results between other classes from the two different settings. We consider it interesting to extend this systematically by studying the status of all further possible inclusions between classes from our hierarchy and classes of the Saluja et al.{} hierarchy. 

\item
We consider it interesting to study systematically the role of built-in relations. E.g., Saluja et al.{} define their classes using only linear order, and prove the hierarchy structure given in Theorem~\ref{thm-saluja-hierarchy}. It can be shown that by adding BIT, SUCC, min and max we obtain $\cPiiRel{1}=\cP$. How does their hierarchy change when we generally introduce SUCC or BIT?
	
\end{asparaenum}




\bibliography{references,cc}

\end{document}